\documentclass[12pt,transaction, draftclsnofoot,onecolumn]{IEEEtran}

\usepackage{amsfonts}
\usepackage{amssymb}
\usepackage{hyperref}
\usepackage{graphicx}
\usepackage{subfigure}
\usepackage{enumerate}
\usepackage{amsmath}
\usepackage{color}
\usepackage{amsthm}
\usepackage{amsmath}
\usepackage{algorithm}
\usepackage{cite}
\usepackage[noend]{algpseudocode}

\newcommand{\bp}{\begin{proof} \small }
\newcommand{\ep}{\end{proof} \normalsize}
\newcommand{\epx}{\end{proof} \small}
\newcommand{\bpa}{\begin{proofappx} \footnotesize }
\newcommand{\epa}{\end{proofappx} \small }
\newtheorem{theorem}{Theorem}
\newtheorem{proposition}{Proposition}
\newtheorem{corollary}{Corollary}

\newtheorem{definition}{Definition}

\newtheorem*{theorem*}{Theorem}
\newtheorem*{proposition*}{Proposition}
\newtheorem*{corollary*}{Corollary}
\newtheorem*{lemma*}{Lemma}
\newtheorem*{assumption*}{Assumption}
\newtheorem*{definition*}{Definition}
\newtheorem*{claim*}{Claim}

\newcommand{\be}{\begin{equation}}
\newcommand{\ee}{\end{equation}}
\newcommand{\bs}{\begin{subequations}}
\newcommand{\es}{\end{subequations}}
\newcommand{\bq}{\begin{eqnarray}}
\newcommand{\eq}{\end{eqnarray}}
\newcommand{\bqn}{\begin{eqnarray*}}
\newcommand{\eqn}{\end{eqnarray*}}

\newcommand{\ba}{\left[ \begin{array}}
\newcommand{\ea}{\\ \end{array} \right]}
\newcommand{\ben}{\begin{enumerate}}
\newcommand{\een}{\end{enumerate}}

\def\a{{\boldsymbol{a}}}

\def\real{{\mathchoice%
{\hbox{\rm\setbox1=\hbox{I}\copy1\kern-.45\wd1 R}}
{\hbox{\rm\setbox1=\hbox{I}\copy1\kern-.45\wd1 R}}
{\hbox{\scriptsize\rm\setbox1=\hbox{I}\copy1\kern-.45\wd1 R}}
{\hbox{\scriptsize\rm\setbox1=\hbox{I}\copy1\kern-.45\wd1 R}}}}

\def\Zint{{\mathchoice{\setbox1=\hbox{\sf Z}\copy1\kern-.75\wd1\box1}
{\setbox1=\hbox{\sf Z}\copy1\kern-.75\wd1\box1}
{\setbox1=\hbox{\scriptsize\sf Z}\copy1\kern-.75\wd1\box1}
{\setbox1=\hbox{\scriptsize\sf Z}\copy1\kern-.75\wd1\box1}}}
\newcommand{\complex}{ \hbox{\rm C\kern-0.45em\rule[.07em]{.02em}{.58em}%
\kern 0.43em}}

\begin{document}
\title{Stochastic Epidemic Networks with \\Strategic Link Formation}

\author{Jie~Xu,~\IEEEmembership{Member,~IEEE}
\thanks{The author is with the Department of Electrical and Computer Engineering, University of Miami, Coral Gables, FL 33146. (Email: jiexu@miami.edu)}% <-this % stops a space
}

\maketitle

\begin{abstract}
Understanding cascading failures or epidemics in networks is crucial for developing effective defensive mechanisms for many critical systems and infrastructures (e.g. biological, social and cyber networks). Most of the existing works treat the network topology as being exogenously given and study under what conditions an epidemic breaks out and/or extinguishes. However, if agents are able to strategically decide their connections according to their own self-interest, the network will instead be endogenously formed and evolving. In such systems, the epidemic, agents' strategic decisions and the network structure become complexly coupled and co-evolve. As a result, existing knowledge may no longer be applicable. Built on a continuous time Susceptible-Infected-Susceptible epidemic model with strong mixing, this paper studies stochastic epidemic networks consisting of strategic agents, who decide the number of links to form based on a careful evaluation of its current obtainable benefit and the potential future cost due to infection by forming links. A game theoretical framework is developed to analyze such networks and a number of important insights are obtained. One key result is that whereas an epidemic eventually dies out if the effective spreading rate is sufficiently low in exogenously given networks, it never dies out when agents are strategic regardless of the effective spreading rate. This property leads to reduced achievable system efficiency and considerably different optimal protection mechanisms. Without understanding the strategic behavior of agents, significant security cost may incur.
\end{abstract}


% \fi
\IEEEpeerreviewmaketitle

\section{Introduction}
Much of society is organized in networks (biological, social, economic, cyber etc.) and networks are important because individuals gain benefits (e.g. information, resource, pleasure etc.) through the interactions with their connected neighbors. But networks also face many vulnerabilities which otherwise would not exist if individuals were isolated. One of the most significant perhaps is cascading failures or epidemics --- viruses (e.g. diseases in human society and computer viruses on the Internet) may spread over the network and damage the society via the links between connected individuals.

There has been a significant effort in the literature \cite{kermack1927contribution, kephart1991directed, wang2003epidemic, pastor2001epidemic, hethcote2000mathematics} devoted to understanding how and under what conditions viruses become epidemic, which is crucial for designing effective protection mechanisms to prevent damage. Most works in this literature study epidemics under a common assumption, namely the networks are \textit{exogenously} determined (according to either a fixed topology or a fixed degree distribution). However, much less knowledge exists for networks that are \textit{endogenously} formed by strategic agents who can decide their connections according to their own interests. In the presence of strategic agents, connection decisions and epidemics are interdependent and hence, tremendously different results regarding, e.g. how viruses become epidemic and how to design effective protection mechanisms, may emerge. The objective of this paper is to advance knowledge in this regard.

This paper builds on the widely-adopted Suspectable-Infected-Suspectable (SIS) model \cite{kephart1991directed} with strong mixing in a continuous time setting. The most prominent feature of our model is that agents participating in the network are strategically deciding how many links to form at each time based on a careful cost-benefit evaluation, rather than simply interact with others over an exogenously given network. Forming links have two opposite effects on the agent: on one hand, forming more links creates more benefits to the agent (due to e.g. more information, resource and pleasure). On the other hand, forming more links exposes the agent to a greater infection threat, which may cause potential future costs to the agent. What makes the problem challenging is that the level of infection threat depends not only on the agent's own decision but also on all other agents' connection decisions since an epidemic is network-wide. In this paper, we develop a game-theoretical framework to model and analyze such strategic behavior of agents in epidemic networks and the implications on optimal network protection. In the considered scenario, agents are continuously interacting with each other and facing stochastic infection threats. Therefore, agents are assumed to be {\it foresighted} rather than myopic. As a result, the conventional equilibrium concept for one-shot games (i.e. Nash equilibrium) is no longer adequate. In this paper, we adopt the conjectural equilibrium \cite{su2010dynamic} as the solution concept. Our model does not exogenously restrict what networks can emerge and does not pose connection constraints. However, the resulting network will depend on specific assumptions on the connection benefit function. Although our model is stylized, many important insights can be obtained for designing effective protection mechanisms against epidemics:
\begin{itemize}
  \item A fundamental result in the literature studying exogenously given epidemic networks is that there is a critical effective spreading rate (i.e. the virus infection rate divided by the curing rate) below which an epidemic eventually extinguishes. This result no longer holds when networks are endogenously formed by strategic agents: for any effective spreading rate, an epidemic never dies out.
  \item Suppose that the network operator can protect the network by choosing to immunize (a fraction of) the agents (at a cost). When networks are exogenously given, immunizing all agents is never the optimal strategy even if the immunization cost is low because it is not necessary. However, when networks are endogenously formed, immunizing all agents indeed will be the optimal strategy when the immunization cost is sufficiently low.
  \item More importantly, the optimal level of immunization protection has a much more complex dependency on the immunization cost in the strategic case than the non-strategic case. Ignoring the strategic link formation behavior of agents may result in sub-optimal immunization deployment, thereby significantly increasing the total network cost caused by the epidemic.
  \item The strategic behavior of agents reduces the achievable system efficiency (defined as the expected utility per capita) and how much the efficiency is reduced depends on the shape of the benefit function of connections. A detailed price of anarchy analysis for the game is provided.
\end{itemize}

The remainder of this paper is organized as follows. In Section II, we discuss related work and highlight the contribution of this paper. Section III develops the system model and introduces the solution concept. In Section IV, we study epidemics in exogenously given networks for comparison purposes. In Section V, we study epidemics in endogenously formed networks. Section VI extends the basic model for homogeneous agents to heterogeneous agents. Numerical results are provided in Section VII. Section VIII concludes the paper and points out future work directions.

\section{Related Work}
Epidemic models can be traced back to McKendrick and Kermack \cite{kermack1927contribution}, which studies a Susceptible-Infected-Recovered (SIR) epidemic model. Since then various mathematical models have been developed for studying epidemics in biological networks and computer networks. The current paper studies the SIS epidemic model \cite{kephart1991directed, pastor2007evolution, bailey1975mathematical, wang2003epidemic, ganesh2005effect, van2009virus}, which is another standard stochastic model for virus infections. In SIS models, individuals are either susceptible or infected: susceptible individuals can be infected and once they have cured/recovered, they are immediately susceptible. Existing works in this regard can be roughly divided into two categories. The first category \cite{kephart1991directed, pastor2007evolution, bailey1975mathematical} adopts a mean-field approximation to study networks consisting of a large number of agents. Recently, the second category \cite{wang2003epidemic, ganesh2005effect, van2009virus} tries to understand the influence of graph characteristics on epidemic spreading rather than a simple homogeneous system. Despite that different mathematical formalisms and tools are adopted, a common fundamental result derived by these models is that there exists a critical effective spreading rate (i.e. the virus infection rate divided by the curing rate) below which the epidemic dies out. For the first category, the critical threshold depends on the degree of the network; for the second category, the critical threshold depends on the largest eigenvalue of the adjacent matrix of the interaction graph. All these works study networks that are exogenously given and agents cannot change the network structure. The present paper significantly departs from this strand of literature and studies strategic agents who form the network endogenously.

Game theoretical frameworks were developed before to study the strategic behavior of agents in epidemic networks. However, the strategic behavior is of a different kind. In these works, agents decide their investment in self-protection, which has an effect in either reducing their individual infection rate \cite{lelarge2009economics} or improving their individual curing rate \cite{omic2009protecting}. More broadly, there is a strand of literature studying individual agent's investment decision in network security games \cite{kunreuther2003interdependent, kearns2003algorithms, laszka2012survey, xu2013rating}, taking into account the interdependent security risks induced by the interaction network. However, again, the networks in this literature are assumed to be fixed.

Endogenously formed networks by strategic agents are studied under the framework of network formation games \cite{bala2000noncooperative, jackson2005survey, jackson2008social}. In most of these works, agents form links only because they can gain benefits from the links (minus constant link cost). Only until recently do researchers begin to study strategic network formation with adversarial attacks and contagious risks \cite{blume2011network, goyal2015strategic}. For instance, \cite{goyal2015strategic} characterizes equilibrium network topologies for games in which agents decide their connections and whether or not to immunize themselves to protect against attack. The current paper studies a very different problem and adopts a very different model. First, we consider a large network and hence adopt a mean-field model. Therefore, our results are in terms of network degrees rather than the specific graph. Second, we consider repeated interactions and hence adopt a continuous time model whereas in \cite{goyal2015strategic}, agents make the decision only once before the attacker moves. Thirdly, we consider a stochastic epidemic setting in which agents are randomly infected and randomly cured. In contrast, \cite{goyal2015strategic} adopts a deterministic infection model --- non-immunized agents reachable to the attacker are attacked and will not be recovered.

In the biological epidemiology research, it is increasingly recognized that a key component of successful infection control effort is understanding the complex, two-way interaction between disease dynamics and human behavioral and social dynamics \cite{wang2015coupled}. Early works study optimal vaccination policies \cite{bauch2004vaccination, bauch2003group}. Recently, economic and game theoretical models have been developed to understand infectious disease transmission when people can engage in public avoidance/social distancing/contact precautions. \cite{reluga2010game} studies a SIR epidemic model with strong mixing and uses a differential game to study the potential value of social distancing as a mitigation measure. The results are supported by numerical results but a theoretical characterization is missing. \cite{chen2011public} studies an SIS model in the discrete time setting with strong mixing and later \cite{chen2012mathematical}  studies the SIR counterpart. Different from our model, agents in these works are solving a myopic best response in each time slot due to the introduction of an explicit infection cost function. In our paper, infection cost is embedded in the future inability to gain benefit and thus, agents are making foresighted decisions. Moreover, whereas \cite{chen2011public, chen2012mathematical} mostly focus on equilibrium/steady state characterization, we explicitly study the optimal protection mechanism as well as the efficiency loss due to the strategic behavior.

\section{Model}
\subsection{Setting}
We consider a large number of agents interacting over a continuous infinite time horizon system. At any $t$, each agent is in one of the two states $s^t \in S =\{H, I\}$ where ``H'' stands for ``Healthy'' and ``I'' stands for ``Infected''. The system state, denoted by $\theta^t$, is the fraction of infected agents at time $t$. We will investigate both cases that the system state $\theta^t$ is known and not known to individual agents. However, we assume that the state of an agent is private information and hence is not known to any other agent in the network at any time. For now, we assume that agents are homogeneous. We will discuss heterogenous agents in Section VI.

\textbf{Healthy agent} At any time $t$, a healthy agent $i$ can choose to form a directed link $i \to j$ with any other agent $j$ in the network. Forming a link creates benefits to agent $i$ that initiates the link but not agent $j$. This resembles the widely adopted followers model (e.g. Twitter).  Since agents are homogeneous and agent states are private information, from agent $i$'s perspective, every other agent $j\neq i$ appears to be the same. Thus, what matters to agent $i$ is the number of links that it forms and the benefit that agent $i$ can obtain only depends on the number of links $a^t_i \in \mathbb{Z}_+$ that it forms. To simplify our analysis, instead of assuming that the number of links can take only integer numbers, we will allow $a^t_i$ to take any non-negative real value in $[0, \infty)$. A non-integer value $a^t_i$ for the number of links can be interpreted as a mixed strategy between $\lfloor a^t_i\rfloor$ and $\lceil a^t_i \rceil$. The real-valued action $a^t_i$ can also be interpreted as the social activity level of the agent.

Forming links creates benefits. Let $b(x)$ denote the instantaneous benefit function depending on the link formation action $x$. This means that the benefit created in a time interval $\Delta t$ is $b(x) \Delta t$ if $x$ links are formed. We assume that $b(x)$ is continuous and twice-differentiable with $b(0) = 0$, $b'(x) > 0$, $b''(x) < 0$. These are standard assumptions which state that the benefit function is increasing with diminishing return. On the other hand, forming links also incurs direct costs. We assume a linear cost function, namely $c(x) = c_0 x$. The instantaneous utility by taking an action $x$ is thus $u(x) = b(x) - c_0 x$. We assume that $b'(0) > c_0$ since otherwise the direct cost is too large for any link to form. With $b'(0) > c_0$, it is easy to see that there exists an optimal action $x$, denoted by $W$, that maximizes the utility $u(x)$. Moreover, we have $u'(x) > 0, \forall x \in [0, W)$,  $u'(W) = 0$ and $u''(x) < 0, \forall x \in [0, W]$.

\textbf{Infected agent} At any time $t$, an infected agent cannot take any active action (meaning that it cannot form links with others and obtain utility). This is a natural situation when, for example, the virus has taken over control of this agent or the agent has broken down. The infected agent will instead go through a curing process. Once the infected agent is cured and become healthy again, it can form links with others. The curing period follows a negative exponential distribution with parameter $\delta > 0$ and the agent losses any utility it can potentially obtain during this period. Note that in the considered model,  the cost due to infection is not explicit; rather, it is implicit as a result of the inability to receive utility. This differs from previous works which assume explicit infection costs. We believe that this is a valid and more natural formulation especially when the direct curing cost is much less than the utility that the agent would have obtained if it were healthy.

\textbf{Infection} Since a healthy agent may initiate to form a link with an infected agent (because the healthy agent does not know the state of any other node), it may get infected via the link that it forms. We model the infection on each link (with an infected node) as a Poison process with arrival rate $\beta$. At the first arrival instance, the healthy agent gets infected. The infection may be sooner (or the infection probability is higher) if the healthy agent forms more links to infected agents. In particular, if a healthy agent forms links with $y$ infected agent, the infection arrival rate is $y\beta$.

\subsection{Strategy and long-term utility}
The strategy of an agent $i$ is the number of links $a^t_i$ to form  at any time $t$ when it is in the healthy state. We start with the case that agents do not know the system state $\theta^t$. In this case, each agent $i$ has the same information set at any time and hence, it adopts a constant strategy $a^t_i = a_i, \forall t: s^t_i = H$. We will study the case that agents observe the system state $\theta^t$ in Section V.B.

If agents are myopic and care about only the instantaneous utility, then it is obvious that the optimal strategy of any healthy agent is $a_i = W$ at any time since doing so maximizes the current time (myopic) utility. However, since agents are interacting with each other for a long time period, they will instead be foresighted and care about the \textit{expected long-term utility}. It is noteworthy that the \textit{expected} long-term utility may be different from the \textit{realized} utility but it is the expected long-term utility that is important for decision making.

Next, we define and compute the expected long-term utility if an agent takes a constant strategy $a_i$. The long-term utility of an agent depends on the state that it is in as well as the strategies taken by other agents. Let $U_H(a_i,\a_{-i})$ and $U_I(a_i, \a_{-i})$ denote the long-term utilities when the agent is in the healthy state and infected state, respectively, by taking strategy $a_i$ while the other agents take the strategies $\a_{-i}$. They are defined recursively as follows:
\begin{align}\label{utility}
U_H(a_i, \a_{-i})
=& \int_{t = 0}^\infty (\int_{\tau = 0}^t e^{-\rho \tau} u(a_i) d \tau + e^{-\rho t} U_I) \beta_i(\a)e^{-\beta_i(\a)t} dt \\
U_I(a_i, \a_{-i}) =& \int_{t = 0}^\infty e^{-\rho t} U_H \delta e^{-\delta t} dt
\end{align}
where $\rho > 0$ is the discount rate, $\beta_i(\a)$ is the infection rate of agent $i$ if the strategy profile $\a = \{a_j,\forall j\}$ of all agents is adopted. By convention, we use $\a_{-i}$ to denote the strategies of all agents other than $i$. We elaborate this definition below:

\begin{itemize}
\item Suppose the agent is infected at time $t_0 + t$ where $t_0$ is the time of the decision making, then $\int_{\tau = 0}^t e^{-\rho \tau} u(a_i) d \tau$ is the discounted long-term utility that the agent can receive during the period $[t_0, t_0 + t]$. Notice that $e^{-\rho \tau} \leq 1$ and decreases with $\tau$, which means that the present value of utility is smaller if the same utility is realized at a later time, thereby discounting. Moreover, a larger $\rho$ means that the discounting effect is greater. Using $e^{-\rho \tau}$ to model discounting is the standard way for continuous time systems.
\item $e^{-\rho t} U_I(a_i)$ is the continuation utility that the agent receives once it gets infected. Since the agent is infected at a later time $t_0 + t$, the utility is discounted by $e^{-\rho t}$. The value of this utility is $U_I$ as the agent enters the infection state at this time point.
\item $\beta_i(\a)e^{-\beta_i(\a)t}$ is the probability distribution function of the infection time $t$, which follows a negative exponential distribution of parameter $\beta_i(\a)$ due to our model. In particular, $\beta_i(\a) = \theta_i(\a) a_i \beta$ where $\beta$ is the per-link infection rate and $\theta_i(\a)$ is the probability that agent $i$ connects to an infected node on each link. Notice that this probability is a result of all agents' strategies and hence, agent $i$'s long-term utility depends on not only its own strategy but also others'.
\item When an agent enters the infection state, it receives 0 instantaneous utility. At a later time $t$, which follows the distribution $\delta e^{-\delta t} dt$, it is cured and returns to the healthy state and the continuation utility is $U_H$ discounted by $e^{-\rho t}$.
\end{itemize}

Since we consider a sufficiently large population of agents, we adopt the mean-field approximation that $\theta_i(\a) = \theta_j(\a),\forall i, j$. This approximation can be well justified: $\theta_i(\a)$ and $\theta_j(\a)$ differs only in situations where agent $i$ connects to agent $j$ and/or agent $j$ connects to agent $i$, which happens with a very small probability when there is a large population and agents randomly form links with others. In fact, the approximation is exact if we consider a continuum population of agents because each individual agent becomes infinitesimal. In what follows, we drop the agent index and simply write $\theta(\a)$ for all $i$. Note that we intentionally use the same symbol $\theta$ for the system state, namely the fraction of infected agents, to denote the probability that the formed link connects to an infected agent. This is because, again, these two quantities are equal due to the law of large numbers in the random matching setting.

\subsection{Solution concept}
The long-term utilities having been defined, the link formation game with infections is clear: each agent is a player in the game, which chooses a link formation strategy to maximize its own long-term utility. Nash equilibrium (NE) is the most commonly used solution concept in game theory. However, NE is inadequate in the considered game since the game is not static and the system state evolves over time depending on the strategies chosen by the players. In this paper, we adopt conjectural (Nash) equilibrium (CE), which is a refinement of NE, as the solution concept. CE is defined as follows:
\begin{definition}
A tuple $\langle \a^*, \theta^*\rangle$, where $\a^*$ is a link formation strategy profile and $\theta^*$ is the conjectured system state, is a CE if starting with any $\theta^0 > 0$, the following two conditions are satisfied:
\begin{enumerate}
  \item $\a^*$ is a Nash equilibrium under the conjectured system state, i.e. $a^*_i = \arg\max_{a_i} U_H(a_i, \a^*_{-i}|\theta^*)$
  \item The conjecture is asymptotically correct, i.e. $\lim_{t\to\infty} \theta^t = \theta^*$ when the nodes adopt $\a^*$.
\end{enumerate}
\end{definition}
Condition 1 states that the agents are playing the game as if the system had converged to the conjectured state even though the actual real-time system state may be different from the conjectured state. Condition 2 states that when the agents adopt the strategies derived based on Condition 1, the system indeed will eventually converge to the conjectured state. Therefore, the conjecture is realized. Conditions 1 and 2 together ensure the consistency of the system dynamics and the solution concept.

\section{Epidemics Under A Fixed Strategy}
Before we study the stochastic epidemic network formed by strategic agents, we first analyze the epidemic propagation if the agents are non-strategic and following a fixed strategy. This analysis will be used for the comparison with the strategic case. In particular, we will focus on a symmetric action profile in which agents adopt the same action, i.e. $a_i = a, \forall i$.

\subsection{Stationary distribution}
The state of the system at any time is characterized by the fraction $\theta^t$ of infected agents among all agents. The evolution of $\theta^t$ depends on the strategy adopted by the agents, the infection rate $\beta$, the curing rate $\delta$ as well as the initial state of the system $\theta^0$. It is obvious that if the system starts with an initial state $\theta^0 = 0$, namely there are zero infected agents, then no matter what strategies are adopted by the agents, the system will remain in the state of zero infection. Therefore, it is more interesting to study the non-trivial case where $\theta^0 > 0$. We focus on the impact of agent strategy on the system state in this subsection and henceforth write $\theta^t(a)$ as a function of the symmetric strategy $a$.

The system is said to be stationary if $\theta^t(a)$ becomes time-invariant. The stationary infection distribution reflects how the epidemic evolves in the long run. We simply drop the time superscript in the stationary distribution notation $\theta(a)$. It turns out that in the considered stochastic link formation setting, there exists a critical number of links such that if agents form fewer than the critical number of links, the epidemic eventually extinguishes, i.e. the stationary distribution $\theta(a) = 0$; otherwise, there is a positive infection probability.
\begin{proposition}
For given $\beta$ and $\delta$, there exists $a_c = \frac{\delta}{\beta}$, such that for all $a \leq a_c$, $\theta(a) = 0$, and for all $a > a_c$, $\theta(a) = 1-\frac{\delta}{\beta a}$.
\end{proposition}
\begin{proof}
To study the stationary distribution of the system, we consider a single agent and compute the probability it is in the infected state at any time, which is
\begin{align}
\theta(a) = \frac{\delta^{-1}}{\delta^{-1} + (\theta(a)a\beta)^{-1}}
\end{align}
There are two solutions of the above equation given any finite $a$. One is $\theta(a) = 0$ and the other is $\theta(a) = 1-\frac{\delta}{\beta a}$. Notice that only when $a > \frac{\delta}{\beta} = a_c$, the second solution is positive. The positive solution gives the probability of nodes being in the infected state and it is differentiable to a threshold $a_c$, where two solutions meet. At the threshold, the solution is not differentiable. For $a < a_c$, the only meaningful solution is $\theta(a) = 0$, which means that all nodes are uninfected in the long-run. To see that the system indeed converges, consider the system dynamics given the symmetric strategy $a$. For any $\theta > 0$, the change in $\theta$ in a small time interval $dt$ is
\begin{align}\label{dynamics1}
d\theta = -\theta \delta dt + (1-\theta)\theta a \beta dt = \theta((1-\theta) a\beta - \delta) dt
\end{align}
Clearly, if $a > a_c$, then there is a solution $\theta^* = 1 - \frac{1}{\beta a}$. For $\theta > \theta^*$, $d\theta < 0$ and for $\theta < \theta^*$, $d\theta > 0$. Thus, the system must converge to $\theta^*$. If $a < a_c$, then for any $\theta > 0$, $d\theta < 0$, which means that the system converges to $\theta^* = 0$. \end{proof}

Proposition 1 states that forming more links with other agents increases the infection probability of the system as a whole. Moreover, there is a critical value of the number of links, which is important for preventing infection in the long-run. The result of Proposition 1 indeed is equivalent to the classic result in the SIS literature, which establishes the existence of a critical effective infection rate (i.e. $\beta/\delta$). This is stated in Corollary 1.
\begin{corollary}
For a given network degree $a$, there exists a critical effective infection rate $\upsilon_c = 1/a$ such that for all $\upsilon \leq \upsilon_c$, $\theta(\upsilon) = 0$, and for all $\upsilon > \upsilon_c$, $\theta(\upsilon) = 1 - \frac{1}{\upsilon_c a}$.
\end{corollary}

\textit{Remark}: If $W \leq a_c$, then the epidemic will always die out even if agents are strategically deciding their link formation actions since agents will form no more than $W$ links. This is the trivial case which we will not study. In the remainder of this paper, we assume that $W > a_c$.

\subsection{Optimal protection}
To reduce the infection probability, the network operator may choose to immunize a fraction of agents at an immunization cost. We assume that immunization is perfect, namely immunized agents do not get infected and hence will not infect other agents in the network. Suppose the network operator can choose to immunize a fraction $(1-\eta) \in [0,1]$ of all agents in the network at the cost of $\gamma\cdot(1-\eta)$ where $\gamma$ is the unit time immunization cost. Without loss of generality, the unit cost to the network operator due to infection is assumed to be 1. Hence, if the infection level is $\theta$, then the infection cost is $\theta$. The total cost (in the stationary state) of the network operator is thus $D(a, \eta) = \theta(a,\eta) + \gamma(1-\eta)$ where $\theta$ depends on the fixed strategy $a$ and the immunization strategy $\eta$. Clearly $\theta < \eta$ since at most $\eta$ fraction of agents can be infected. The following proposition characterizes the optimal immunization strategy $\eta^*$ given agents' fixed strategy $a$. We focus on the non-trivial case where $a > a_c$.
\begin{proposition}
If $\gamma \leq 1$, then $\eta^* = \frac{\delta}{\beta a}$. If $\gamma > 1$, then $\eta^* = 1$.
\end{proposition}
\begin{proof}
The stationary infection level $\theta(a,\eta)$ satisfies,
\begin{align}
\theta(a, \eta) = \frac{\eta \delta^{-1}}{\delta^{-1} + (\theta(a,\eta)a\beta)^{-1}}
\end{align}
If $\eta \leq \frac{\delta}{\beta a}$, then $\theta(a,\eta) = 0$. If $\eta > \frac{\delta}{\beta a}$, then $\theta(a, \eta) = \eta - \frac{\delta}{\beta a}$. We analyze these two cases separately.

Case 1: if $\eta \leq \frac{\delta}{\beta a}$, then $D = \gamma(1 - \eta)$. Therefore, the optimal $\eta$ is $\eta^* = \frac{\delta}{\beta a}$ in order to minimize $D$. The minimum total cost is $D = \gamma(1-\frac{\delta}{\beta a})$.

Case 2: if $\eta > \frac{\delta}{\beta a}$, then $D = \eta - \frac{\delta}{\beta a} + \gamma(1-\eta) = \gamma - \frac{\delta}{\beta a} + (1-c)\eta$. If $\gamma > 1$, then $D$ is decreasing in $\eta$ and hence the minimum $D$ is achieved at $\eta = 1$. If $\gamma \leq 1$, then $D$ is increasing in $\eta$ and hence the minimum $D$ is achieved at $\eta = \frac{\delta}{\beta a}$.

In sum, if $\gamma > 1$, then $\eta^* = 1$; if $ \gamma \leq 1$, then $\eta^* = \frac{\delta}{\beta a}$.
\end{proof}

Proposition 2 states that if the immunization is too costly (i.e. $\gamma > 1$), then the optimal strategy is to immunize no agents (i.e. $1-\eta = 0$). If immunization is not too costly, the optimal strategy is to immunize a fraction $(1-\frac{\delta}{\beta a})$  of agents depending on the fixed strategy $a$. Moreover, as long as $\eta < 1$, the optimal strategy does not depend on the immunization cost so even if $\eta$ is small, it is not optimal to immunize all agents. This is because by immunizing a sufficiently large fraction of the agents, the epidemic can already die out. Therefore, it is not necessary to immunize all agents at additional costs.

\subsection{System efficiency and social optimum}
Suppose that instead of deploying an immunization strategy, the network operator can control the link formation strategy of the agents. Obviously, there is a tradeoff when the utility of link formation is involved: forming fewer links reduces the infection level of the system but also decreases the utility that agents can obtain. The objective of the network operator is thus to determine the optimal strategy that maximizes the system efficiency, which is defined as the average utility per capita in the stationary distribution, denoted by $E(a)$.
Specifically, $E(a)$ can be computed as $E(a) = (1-\theta(a))u(a)$. The network operator's problem is to solve $a^{OPT} = \arg\max_a E(a)$. Denote the social optimum by $E^{OPT}$. The proposition below characterizes the optimal solution.
\begin{proposition}
$a^{OPT} = a_c$ and $E^{OPT} = u(a_c)$.
\end{proposition}
\begin{proof}
For any $a \leq a_c$, $\theta(a) = 0$. Therefore $E(a) = u(a), \forall a \leq a_c$. Clearly $E(a)$ is maximized at $a = a_c$. For any $a \geq a_c$, $\theta(a) = 1 - \frac{\delta}{\beta a}$. Therefore $E(a) = \frac{\delta u(a)}{\beta a}$. Take the first-order derivative,
\begin{align}
E'(a) = \frac{\delta}{\beta}\frac{u'(a) a - u(a)}{a^2}
\end{align}
Let $f(a) = u'(a)a -u(a)$. It is easy to check that $f(0) = 0$, $f(\infty) < 0$ and $f(a)$ is always decreasing (because $f'(a) = u''(a) < 0$) and hence, $f(a) < 0, \forall a > 0$. Thus $E'(a) < 0, \forall a > 0$ and hence $E(a)$ is decreasing in $[a_c, \infty)$. Therefore, $E(a)$ is maximized at $a = a_c$. In sum, $E(a)$ is maximized at $a = a_c$ and $E(a_c) = u(a_c)$.
\end{proof}
Proposition 3 states that the threshold number of links eliminates infection in the long run and at the same time also maximizes the system efficiency.

\section{Epidemics with Strategic Agents}
In the previous section, we studied the epidemic propagation, the immunization mechanism and the system efficiency when agents are non-strategic and following a fixed strategy. In this section, we will show that tremendously different results emerge when agents are strategic who try to maximize their own expected long-term utility by deciding the number of links to form.

\subsection{Conjectural equilibrium (no observation of state state)}
In this subsection, we study the case that agents do not observe the system state at any moment in time. Since the information set of agents is the same across time, we use the conjectural equilibrium (CE), defined in Section III.C, as the solution concept. We start with a more explicit expression of the agent's long-term utility. By performing the integral operation on \eqref{utility}, the long-term utilities are simplified to
\begin{align}
U_H(a_i, \a_{-i}) =& \frac{1}{\rho + \beta_i(\a)}\left(u(a) + \beta_i(\a) U_i(a_i, \a_{-i})\right)\\
U_L(a_i, \a_{-i}) =& \frac{\delta}{\rho+\delta}U_L(a_i, \a_{-i})
\end{align}
Substituting $U_L(a_i, \a_{-i})$ into the first equation, we obtain
\begin{align}
U_H(a_i, \a_{-i}) = \frac{\rho + \delta}{\rho}\frac{u(a_i)}{\rho + \delta + \beta\theta(\a)a_i}
\end{align}
In CE, $a^*_i$ is the best response to other agents' equilibrium strategies, i.e. $a^*_i = \arg\max_{a_i} U_H(a_i, \a^*_i)$. Moreover, the conjecture has to be consistent with the reality. Theorem 1 establishes the existence and uniqueness of the CE.
\begin{theorem}
The link formation game admits a unique CE. Moreover, the CE is symmetric, namely $a^{CE}_i = a^{CE},\forall i$ where $a^{CE} > a_c$ is the unique solution of
\begin{align}\label{NE}
\frac{u(a)}{u'(a)} - a = \frac{\rho + \delta}{\beta - \frac{\delta}{a}}
\end{align}
\end{theorem}
\begin{proof}
We consider two cases depending on the value of $\theta$ below.

(1) Suppose in the equilibrium, we have $\theta > 0$. Then, for any $\theta > 0$, we prove that there is a unique maximizer $a^*_i(\theta)$. To prove this, take the first-order derivative on the objective function, which is
\begin{align}
U'_H(a_i) = \frac{\rho+\delta}{\rho}\cdot \frac{u'(a_i)(\rho + \delta + \beta\theta a_i) - u(a_i)\beta\theta}{(\rho + \delta + \beta\theta a_i)^2}
\end{align}
Let the numerator function  be $f(a_i) = u'(a_i)(\rho + \delta + \beta\theta a_i) - u(a_i)\beta\theta$ for the ease of exposition. Clearly $f(0) > 0$ and $f(M) < 0$. Moreover, $f'(a_i) = u''(a_i)(\rho + \delta + \beta\theta a_i) < 0$. Therefore, there exists a unique $a^*_i\in (0, M)$ such that $f'(a_i) > 0, \forall a < a^*_i$, $f'(a^*_i) = 0$ and $f'(a_i) < 0, \forall a > a^*_i$. As a result, there is a unique maximizer $a^*_i\in (0, M)$ of $U_H(a_i)$ for $\theta \in (0,1)$, which is the solution of
\begin{align}
\frac{u(a)}{u'(a)} - a = \frac{\rho + \delta}{\beta\theta}
\end{align}
Because every agent faces the same system state $\theta$, this proves that, if an equilibrium exists, agents' strategies must be symmetric. We henceforth drop the agent index and write the symmetric strategy as $a^*(\theta)$.

Because strategies are symmetric in equilibrium, we can then apply the result of Proposition 1 to compute the stationary distribution. Since we have assumed that $\theta > 0$, equilibrium must satisfy two conditions
\begin{align}
\theta = 1 - \frac{\delta}{\beta a^*(\theta)}~~~~\textrm{and}~~~~a^*(\theta) > a_c
\end{align}
Substituting $\theta$ into the first-order condition, we have
\begin{align}\label{Condition}
\frac{u(a)}{u'(a)} - a = \frac{\rho + \delta}{\beta - \frac{\delta}{a}}
\end{align}
To prove that there is a unique solution to the above equation, consider the monotonicity of both sides of the equation with respect to $a$. Clearly, the right-hand side (LHS) is a decreasing function in $a$ for $a > a_c = \frac{\delta}{\beta}$. The first-order derivative of the left-hand side function is
\begin{align}
\frac{u'(a)u'(a) - u(a)u''(a)}{(u'(a))^2} - 1 = \frac{- u(a)u''(a)}{(u'(a))^2} > 0
\end{align}
Therefore, the right-hand side (RHS) is increasing in $a$. Moreover, notice that as $a \to a_c$, LHS must be smaller than RHS while as $a \to M$, LHS must be greater than RHS. Therefore, there must exist a unique solution $a^* \in (a_c, M)$. This solution satisfies the two conditions of CE.

(2) Suppose in the equilibrium, we have $\theta = 0$. In this case, the utility becomes
\begin{align}
U_H(a_i, \a_{-i})  = \frac{u(a_i)}{\rho}
\end{align}
which monotonically increases with $a_i$. Hence, each agent chooses $a_i$ as large as possible. However, if $a_i$ goes over $a_c$, infection does not extinguish, thereby contradicting the assumption that $\theta = 0$. Therefore, $\theta = 0$ cannot be a CE.
\end{proof}

Theorem 1 states a result that is in stark contrast with existing works on epidemic networks with fixed degrees: when agents are strategic in forming links, infection \textit{never} extinguishes in equilibrium for any value of $\beta$ and $\delta$ (hence any value of the effective spreading rate) because the number of links that the agents are willing to form is always larger than the corresponding critical value $a_c(\beta, \delta)$. Proposition 4 further reveals the dependency of the CE on the system parameters.

\begin{proposition}
$a^{CE}$ increases with $\rho$, $\delta$ and decreases with $\beta$.
\end{proposition}
\begin{proof}
Let $a^{CE}_1$ be the corresponding CE for $\rho_1$ and $a^{CE}_2$ for $\rho_2$ for fixed $\delta$ and $\beta$. Assume $\rho_1 < \rho_2$. Observe \eqref{NE}, we must have $LHS(a^{CE}_k) = RHS(a^{CE}_k, \rho_k), \forall k = 1, 2$. Suppose $a^{CE}_1 \geq a^{CE}_2$. Since LHS is increasing in $a$, we have $LHS(a^{CE}_1) \geq LHS(a^{CE}_2)$. Since RHS is decreasing in $a$ and increasing in $\rho$, we have $RHS(a^{CE}_1, \rho_1) < RHS(a^{CE}_2, \rho_2)$. Therefore, $LHS(a^{CE}_2) < RHS(a^{CE}_2, \rho_2)$, which is a contradiction. Therefore,  $a^{CE}_1 < a^{CE}_2$. This proves that $a^{CE}$ increases with $\rho$. The other two claims can be similarly proved.
\end{proof}

Proposition 4 can be intuitively understood. (1) A larger $\rho$ means greater discounting of future utilities. Therefore, agents have incentives to form more links to reap more instantaneous utilities since future utilities are less valued. (2) A larger $\beta$ means faster infection. Since forming more links exposes the agent to more infection threats when $\beta$ is larger, agents have incentives to reduce the number of links formed. (3) A larger $\delta$ means faster curing. Thus the infection threats become smaller when $\delta$ is larger because the agents can recover more quickly.

\subsection{System dynamics and convergence (observation of the system state)}
In the previous subsection, we studied the case that agents do not observe the actual system state $\theta^t$ at any time $t$. In such situations, conjectural equilibrium is a reasonable solution concept in which agents form conjecture about the converged system state. In this subsection, we study the case that agents observe the system state $\theta^t$ and act upon this observation to maximize their long-term utilities. We are interested in the convergence of the system: whether it converges and, if yes, where it converges to, and how fast it converges.

At any time $t$, each agent observes the current state $\theta^t$ and determines the link formation action $a^t$ based on this information. Due to the assumption of bounded rationality, $a^t$ is the best response to $\theta^t$, i.e. $a^*(\theta^t) = \arg\max_a U_H(a|\theta^t)$, where $U_H(a|\theta^t)$ is
\begin{align}
U_H(a|\theta^t) = \frac{\rho + \delta}{\rho}\frac{\mu(a)}{\rho+\delta +\beta \theta^t a}
\end{align}
which can be derived based on our previous analysis. Note that all healthy agents take the same action $a^*(\theta^t)$ since agents are homogenous and face the same system state. Moreover, $a^*(\theta^t)$ is the solution of
\begin{align}\label{BR}
\frac{u(a^t)}{u'(a^t)} - a^t = \frac{\rho +\delta}{\beta \theta^t}
\end{align}
The system evolution therefore can be characterized by the following differential equation
\begin{align}\label{dynamics2}
d\theta^t = -\theta^t\delta dt + (1-\theta^t)\beta\theta^t a^*(\theta^t) dt
\end{align}
In the above system dynamics equation, the first term $\theta^t \delta dt$ is the population mass of infected agents that are cured in a small time interval $dt$. The second term $(1-\theta^t)\beta\theta^t a^*(\theta^t) dt$ is the population mass of healthy agents that are infected in a small interval $dt$. Notice the difference between \eqref{dynamics1} and \eqref{dynamics2}: in \eqref{dynamics1}, the link formation strategy is fixed over time whereas in \eqref{dynamics2}, the link formation strategy is the best response to the current system state.

\begin{theorem}
The system dynamics defined by \eqref{dynamics2} converges to $\theta^{CE}$ starting from any $\theta^0 > 0$.
\end{theorem}
\begin{proof}
The best response $a^*(\theta^t)$ is determined by equation \eqref{BR}. Because LHS of \eqref{BR} is monotonically increasing with $a$, $a^*(\theta^t)$ is decreasing with $\theta^t$. Now consider the system dynamics. Equation \eqref{dynamics2} can be rewritten as
\begin{align}
d\theta^t = \theta^t [(1-\theta^t)\beta a^*(\theta^t) - \delta] dt
\end{align}
We are interested in the sign of $d\theta^t$ for different values of $\theta^t$. Since $\theta^t \geq 0$, what matters is $f(\theta^t) \triangleq (1-\theta^t)\beta a^*(\theta^t) - \delta$. Since $a^*(\theta^t)$ is decreasing in $\theta^t$, $f(\theta^t)$ is decreasing in $\theta^t$. Next, we show that there exists $\theta^* \in (0,1)$ such that $f(\theta^*) = 0$.

$\theta^*$ must satisfy the following set of equations
\begin{align}
(1-\theta^*)\beta a^* - \delta = 0,~~~\textrm{and}~~~
\frac{u(a^*)}{u'(a^*)} - a^* = \frac{\rho + \delta}{\beta\theta^*}
\end{align}
Rearranging the first equation to get $\theta^*$ in terms of $a^*$ and substituting $\theta^*$ in the second equation gives
\begin{align}
\frac{u(a^*)}{u'(a^*)} - a^* = \frac{\rho + \delta}{\beta - \frac{\delta}{a^*}}
\end{align}
Notice that the above is the same equation of determining the CE $a^{CE}$. Therefore $a^* > a_c$ and hence $\theta^* \in (0,1)$. Since $f(\theta^*) = 0$ and $f(\theta^t)$ is decreasing, we know that $f(\theta^t) < 0$ for $\theta^t \in (\theta^*, 1)$ and $f(\theta^t) > 0$ for $\theta^t \in (0, \theta^*)$. This means that $\theta^t$ decreases when $\theta^t \in (\theta^*, 1)$ and increases when $\theta^t \in (0, \theta^*)$. Therefore, the system converges to the unique steady state $\theta^* = \theta^{CE}$.
\end{proof}

Theorem 2 states that the best response dynamics leads the system to a converged state which corresponds to the conjectural equilibrium. It is worth highlighting the differences between CE and best response again: in CE, every agent does not observe the system state at any time and acts based on the conjectured system state. As a result, the link formation strategy does not change over time. In teh best response dynamics, every agent observes the (same) system state and optimizes its strategy against the current system state. Therefore, the link formation strategy evolves over time. Nevertheless, Theorem 2 proves that, with or without the knowledge of the system state, the system converges to the same state in which infection never extinguishes, which is significantly different from a system where agents are obedient and interact on fixed network topologies.

The proposition below provides bounds on the rate of convergence.
\begin{proposition}
For any initial system state $\theta^0 > \theta^{CE}$ (or $\theta^0 < \theta^{CE})$, $\forall \epsilon > 0$, let $T(\theta^0, \epsilon)$ be the time at which $\theta^T$ decreases to $\theta^{CE} + \epsilon \triangleq \theta^\epsilon$ (or increases to $\theta^{CE} - \epsilon \triangleq \theta^\epsilon$), we have
\begin{align}
\frac{\ln \theta^0 - \ln \theta^\epsilon}{(1-\theta^0)\beta a^*(\theta^0) - \delta} < T(\theta^0, \epsilon) < \frac{\ln \theta^0 - \ln \theta^\epsilon}{(1-\theta^\epsilon)\beta a^*(\theta^\epsilon) - \delta}
\end{align}
\end{proposition}
\begin{proof}
The system dynamics can be written as
\begin{align}
d\ln \theta^t=\frac{d\theta^t}{\theta^t} =  [(1-\theta^t)\beta a^*(\theta^t) - \delta] dt
\end{align}
Since $\ln \theta^t$ is monotonically increasing in $(0, 1)$, $\theta^0$ evolving to $\theta^\epsilon$ is equivalent to $\ln \theta^0$ evolving to $\ln \theta^\epsilon$. Because we have shown that $(1-\theta^t)\beta a^*(\theta^t) - \delta$ is decreasing in $\theta^t$ and $(1-\theta^{CE})\beta a^*(\theta^{CE}) - \delta = 0$, the absolute rate of change $|(1-\theta^t)\beta a^*(\theta^t) - \delta|$ is larger if $\theta^t$ is further away from $\theta^{CE}$. Therefore, before $\ln \theta^t$ decreases (or increases) to $\ln \theta^\epsilon$, the rate of change is at least $(1-\theta^\epsilon)\beta a^*(\theta^\epsilon) - \delta$ and at most $(1-\theta^0)\beta a^*(\theta^0) - \delta$. This proves the proposition.
\end{proof}

\subsection{Optimal protection}
Now, we study the optimal immunization strategy for epidemic networks with strategic agents. Again, the network operator chooses to immunize a fraction $(1-\eta)\in[0,1]$ of the agents at a cost of $\gamma\cdot (1-\eta)$. Different from the case where agents adopt a fixed strategy, the immunization strategy of the network operator now influences how the strategic agents choose their link formation action and hence the total cost. The total cost is therefore $D(\eta) = \theta(a(\eta), \eta) + \gamma(1-\eta)$. The next proposition characterizes the impact of the immunization strategy on the total cost.
\begin{proposition}
Assume $u'''(a) < 0$ and $W \to \infty$. There exist $\gamma_1, \gamma_2~(\gamma_1 < \gamma_2)$ such that
\begin{itemize}
  \item If $\gamma \leq \gamma_1$, then $D(\eta)$ is increasing in $\eta$.
  \item If $\gamma \geq \gamma_2$, then $D(\eta)$ is decreasing in $\eta$.
  \item If $\gamma \in (\gamma_1, \gamma_2)$, then there exists $\eta^*(\gamma)$ such that $D(\eta)$ is decreasing in $[0, \eta^*(\gamma)]$ and then increasing in $[\eta^*(\gamma), 1]$. Moreover, $\eta^*(\gamma)$ is increasing in $c$.
\end{itemize}
\end{proposition}
\begin{proof}
The infection level $\theta$ depends on $\eta$ as follows: $\theta = \eta - \frac{\delta}{\beta a^*(\theta)}$ where $a^*(\theta)$ satisfies
\begin{align}
f(a) \triangleq \frac{u(a)}{u'(a)} - a = \frac{\rho +\delta}{\beta\theta} =  \frac{\rho+\delta}{\eta\beta - \frac{\delta}{a}}
\end{align}
Using similar arguments as in Proposition 4, we can show that $a^*$ is decreasing in $\eta$. Because $a^*$ is decreasing in $\theta$, we have that $\theta$ is monotonically increasing in $\eta$, i.e. $\theta'(\eta) > 0, \forall \eta$. Next, we prove that $\theta''(\eta) > 0$. Because (1) $\theta(0) = 0$ and (2) $\theta$ and $\eta$ is one-to-one mapping, it is equivalent to prove that $\eta''(\theta) < 0$ where $\eta = \theta + \frac{\delta}{\beta a^*(\theta)}$. Thus, it further is equivalent to prove that $\forall \theta_2 > \theta_1$, the following is true
\begin{align}
\frac{\theta_2}{\theta_1} > \frac{\frac{1}{a^*(\theta_2)}}{\frac{1}{a^*(\theta_1)}} = \frac{a^*_1}{a^*_2}
\end{align}
where $a^*_1 > a^*_2$.  Since $f(a) = \frac{\rho+\delta}{\beta\theta}$, we can instead prove
\begin{align}
\frac{a^*_1}{a^*_2} < \frac{f(a^*_1)}{f(a^*_2)}
\end{align}
Because $f(0) = 0$ and $f'(a) > 0$, it is equivalent to prove that $f''(a) > 0$. To see this is true,
\begin{align}
f'(a) = \frac{u'(a)u'(a) - u(a) u''(a)}{(u'(a))^2} - 1 = -\frac{u(a)u''(a)}{(u'(a))^2} > 0
\end{align}
and
\begin{align}
f''(a) = -\frac{(u'(a)u''(a) + u(a)u'''(a))u'(a) - 2 u(a)(u''(a))^2}{(u'(a))^3}
\end{align}
Therefore, if $u'''(a) < 0$, then $f''(a) > 0$. This proves that $\theta''(\eta) > 0$. Hence, $\theta'(\eta) > 0$ is an increasing function of $\eta$. Let $\lim_{\eta \to 0}\theta'(\eta) = \gamma_1$ and $\lim_{\eta \to 1}\theta'(\eta) = \gamma_2$. Clearly $\gamma_1 < \gamma_2$.

Notice $D'(\eta) = \theta'(\eta) - \gamma$. If $\gamma < \gamma_1$, then $D'(\eta) > 0, \forall \eta$. This means that $D(\eta)$ is increasing in $\eta$. If $\gamma > \gamma_2$, then $D'(\eta) < 0, \forall \eta$. This means that $D(\eta)$ is decreasing in $\eta$. If $\gamma \in (\gamma_1, \gamma_2)$, then there must exist $\eta^*$ such that $\theta'(\eta^*) = \gamma$. This means that $D(\eta)$ is decreasing in $[0, \eta^*]$ and then increasing in $[\eta^*, 1]$. Finally, it is easy to see that $\eta^*(\gamma)$ is increasing in $\gamma$.
\end{proof}
Proposition 6 states an important result regarding the optimal protection deployment. Specifically, the optimal immunization strategy depends on the relative cost of immunization and can be categorized into three regions. If the immunization cost is sufficiently low (i.e. $\gamma \leq \gamma_1$), then the optimal strategy is immunizing all agents in the network. If the immunization cost is too high (i.e. $\gamma \geq \gamma_2$), then the optimal strategy is immunizing no agents in the network. If the immunization cost is neither too low nor too high, there is a unique optimal level of protection $1-\eta^*(\gamma)$ that depends on the cost $\gamma$.

The optimal protection when agents are strategic is significantly different from that when agents follow a fixed link formation strategy. In the latter case, no matter what fixed strategies the agents adopt and what the immunization cost is, immunizing all agents is never the optimal protection from the network operator's perspective. This is because immunizing all agents is \textit{unnecessary} since as long as the immunization is sufficiently strong, infection will extinguish. However, in the former case, immunizing all agents is indeed the optimal protection strategy when the cost is low. This is because, as shown in Theorem 1, no matter how strong the immunization is, strategic agents have incentives to form more than the critical number of links and hence, infection never extinguishes. Therefore, if the immunization cost is sufficiently small, then it is optimal for the network operator to proactively eliminate infection through immunizing all agents. In Section VII, we will further elaborate these points through numerical results and show that ignoring the strategic nature of agents may lead to significantly higher system costs.

\subsection{Price of anarchy}
In this subsection, we study the system efficiency when agents are strategic. The price of anarchy (POA) is an important game theoretical concept that measures how the system efficiency degrades due to selfish behavior of the agents. PoA is defined as the ratio between the social optimum and the worst equilibrium efficiency. Hence, PoA is always no less than 1 and the larger PoA, the larger efficiency loss. Since there is a unique equilibrium in our problem, PoA is simply $PoA = E^{OPT}/E^{CE}$. Recall that for a given link formation strategy $a$, the system efficiency is defined as $E(a) = (1-\theta(a))u(a) = \frac{\delta u(a)}{\beta a}$.

Before we proceed, we note that there is a trivial bound on PoA: since it must be $a^{CE}\in (a_c, W)$ and $E(a)$ is decreasing in $[a_c, W]$ according to the proof of Proposition 3, the PoA must satisfy
\begin{align}
PoA < \frac{E^{OPT}}{E(W)}
\end{align}
If $W$ is close to $a_c$, then the PoA bound is already very tight. In the next proposition, we consider the case when $W$ is at least twice of $a_c$.

\begin{proposition}
Let $a^\dagger \triangleq \frac{\delta + \sqrt{\delta^2 + \rho\delta}}{\beta} > 2a_c$  and $\kappa \triangleq \frac{a^\dagger u(a_c)}{a_c u(a^\dagger)} > 1$. If $W > a^\dagger$, then the following holds.
\begin{itemize}
  \item If $u'(a^\dagger) < \frac{u(a^\dagger)}{2a^\dagger + \rho \beta^{-1}}$, then $PoA < \kappa$.
  \item If $u'(a^\dagger) > \frac{u(a^\dagger)}{2a^\dagger + \rho \beta^{-1}}$, then $PoA > \kappa$.
  \item If $u'(a^\dagger) = \frac{u(a^\dagger)}{2a^\dagger + \rho \beta^{-1}}$, then $PoA = \kappa$.
\end{itemize}
\end{proposition}
\begin{proof}
Consider \eqref{NE}, we have already proven that $a^{CE} > a_c$ is the unique solution to it. Rearranging \eqref{NE} gives
\begin{align}\label{NewCondition}
g(a) \triangleq \frac{u(a)}{u'(a)} = \frac{\rho + \delta}{\beta - \frac{\delta}{a}} + a = \frac{\beta a^2 + \rho a}{\beta a - \delta} \triangleq f(a)
\end{align}
Since \eqref{NewCondition} is equivalent to \eqref{Condition}, there must be a unique solution $a > a_c$ such that $g(a) = f(a)$.

First, we prove that $g(a)$ is increasing in the range of $[a_c, M)$. This is easy to see because $g'(a) = \frac{(u'(a))^2 - u(a)u''(a)}{(u'(a))^2} > 0$. Next, we investigate the monotonicity of $f(a)$. Taking the first-order derivative,
\begin{align}
f'(a) = \frac{\beta^2 a^2 - 2\beta\delta a - \rho \delta}{(\beta a - \delta)^2}
= \frac{1}{(\beta a - \delta)^2}\left(a - (a_c + \phi_{\rho,\delta,\beta})\right)\left(a - (a_c - \phi_{\rho,\delta,\beta})\right)
\end{align}
where $\phi_{\rho,\delta,\beta} \triangleq \frac{\sqrt{\delta^2 + \rho\delta}}{\beta}$ is a constant depending on the system parameters. Let $a^\dagger = a_c + \phi_{\rho,\delta,\beta}$. Therefore $f'(a) < 0$ for $a\in (a_c, a^\dagger)$ and $f'(a) > 0$, for $a \in (a^\dagger, \infty)$. This means that $f(a)$ is decreasing in $a\in (a_c, a^\dagger)$ and increasing in $a \in (a^\dagger, \infty)$. Moreover,
\begin{align}
f(a^\dagger) = \frac{2\delta + \rho + 2\sqrt{\delta^2 + \rho\delta}}{\beta} = 2a^\dagger + \rho\beta^{-1}
\end{align}

Since $a^{CE}$ is at where $g(a)$ and $f(a)$ intersects, the value of $a^{CE}$ depends on whether $g(a^\dagger) > f(a^\dagger)$ or $g(a^\dagger) < f(a^\dagger)$.
If $g(a^\dagger) > f(a^\dagger)$, then it must be $a^{CE} \in (a_c, a^\dagger)$. If $g(a^\dagger) < f(a^\dagger)$, then it must be $a^{CE} \in (a^\dagger, W)$. If $g(a^\dagger) = f(a^\dagger)$, then $a^{CE} = a^\dagger$. Considering $E(a)$ is decreasing in $a \in (a_c, W)$, the proposition is proved.
\end{proof}
Proposition 7 shows how close $a^{CE}$ is to $a_c$ depends on the specific form of the utility function, specifically, the slope of the utility function at a particular point $a^\dagger$. If the utility function increases slowly enough after forming a certain number of $a^\dagger$ links, then agents have incentives to form less than $a^\dagger$ links in equilibrium. This is intuitive because the increased infection threat outweighs the increased utility by forming more links. Otherwise, agents have incentives to form more than $a^\dagger$ links in equilibrium. Moreover, any utility function that satisfies $u'(a^\dagger) = \frac{u(a^\dagger)}{2a^\dagger + \rho \beta^{-1}}$ must have the equilibrium be exactly $a^\dagger$, regardless of the value of $u(x)$ at other points than $a^\dagger$. This analysis leads to the bounds on the PoA: the PoA is upper bounded (i.e. efficiency loss is lower bounded) by a constant if the instantaneous utility of forming more links increases sufficiently slowly whereas PoA is lower bounded (i.e. efficiency loss is upper bounded) by a constant if the instantaneous utility of forming more links increases sufficiently fast.

\section{Heterogenous Agents}
In the previous sections, we assumed that agents are homogenous. Many of our results can be extended to the case of heterogenous agents. In this section, we consider a heterogenous system where agents are different in terms of the curing rates and show that conjectural equilibrium exists in this case. In particular, we assume that there are $K$ types of agents, indexed by the set $\mathcal{K} = \{1,2,...,K\}$. The curing rate for type $k$ agents is $\delta_k$ and the fraction of type $k$ agents is $w_k$ and $\sum_{k\in \mathcal{K}} w_k = 1$.

First, we characterize the stationary distribution when the link formation strategy $(a_1, ..., a_K)$ is adopted in which all type $k$ agents adopt the constant strategy $a_k$.
\begin{proposition}
For given $\beta$ and the vector of curing rates $(\delta_1, ..., \delta_K)$, there exists a convex set $\mathcal{C} = \{(a_1,...,a_K): \beta \sum_k w_k \delta^{-1}_k a_k \leq 1\}$, such that if $(a_1,...,a_K) \in \mathcal{C}$, then $\theta(a_1, ..., a_K) = 0$. Otherwise, there is a positive infection probability $\theta$, which is the unique solution to
\begin{align}\label{SDheter}
\sum_k w_k \frac{ \beta \delta^{-1}_k a_k}{\theta  \beta \delta^{-1}_k a_k + 1} = 1
\end{align}
\end{proposition}
\begin{proof}
Let $\theta_k$ be the fraction of infected agents among all type $k$ agents. In the steady state, we have
\begin{align}\label{heter1}
\theta_k = \frac{\delta^{-1}_k}{\delta^{-1} + (\theta a_k \beta)^{-1}} = \frac{\theta a_k \beta \delta^{-1}_k}{\theta a_k \beta\delta^{-1}_k + 1}, \forall k = 1, 2, ..., K
\end{align}
where the fraction of infected agents among all agents is $\theta = \sum_k w_k \theta_k$. It is clear that if $\theta > 0$, then $\theta_k > 0, \forall k$ and if $\theta = 0$, then $\theta_k = 0, \forall k$. $\theta = 0$ is a trivial solution of the above equation in which $a_k, \forall k$ can be any value. We now study the non-trivial solution $\theta > 0$. From the above equation, we have
\begin{align}
\theta = \frac{\theta_k}{a_k\beta\delta^{-1}_k (1 - \theta_k)}, \forall k = 1,2,...,K
\end{align}
and hence, $\theta_k = a_k\beta\delta^{-1}_k (1 - \theta_k) \theta, \forall k$. Summing up both sides over $k$, multiplied by $w_k$, we have
\begin{align}
\theta = \sum\limits_{k}w_k \theta_k = \sum\limits_{k} w_k a_k\beta\delta^{-1}_k (1 - \theta_k) \theta
\end{align}
This leads to
\begin{align}\label{heter2}
\sum\limits_{k} w_k a_k\beta\delta^{-1}_k (1 - \theta_k) = 1
\end{align}
For all $\theta > 0$, we must have $\sum\limits_{k} w_k a_k\beta\delta^{-1}_k > 1$. This means that if $\sum\limits_{k} w_k a_k\beta\delta^{-1}_k \leq 1$, the only solution is $\theta = 0$.

We also need to show that if $\sum\limits_{k} w_k a_k\beta\delta^{-1}_k > 1$, there indeed exists a unique solution $\theta > 0$. Substituting \eqref{heter1} into  \eqref{heter2} gives
\begin{align}
f(\theta) \triangleq \sum_k w_k \frac{a_k \beta \delta^{-1}_k}{\theta a_k \beta \delta^{-1}_k + 1} = 1
\end{align}
Clearly $f(\theta)$ is decreasing in $\theta$. Moreover, $f(0) = \sum_{k} w_k a_k\beta\delta^{-1}_k > 1$ and $f(1) < \sum_k w_k = 1$. Therefore, there must be a unique solution $\theta \in (0,1)$ such that $f(\theta) = 1$.
\end{proof}
Recall that when agents are homogeneous, there is a single threshold number of links that is important for determining whether infection persists or extinguishes (see Proposition 1). In the case of heterogeneous agents, the result is analogous. Since the link formation action is characterized by a vector instead of a scalar value, the threshold becomes a hyperplane. It is easy to see that when $K = 1$, Proposition 7 reduces to Proposition 1. Corollary 2 below is the extension to Corollary 1, which again resembles the classic result in the SIS literature.

\begin{corollary}
For a given vector of curing rates $(\delta_1, ..., \delta_K)$ and a vector of link formation strategy $(a_1, ..., a_K)$, there exists $\beta_c = \frac{1}{\sum_k w_k\delta^{-1}_k a_k}$ such that for all $\beta \leq \beta_c$, $\theta(\beta) = 0$, and for all $\beta > \beta_c$, $\theta(\beta) > 0$.
\end{corollary}

Now we are able to show that CE exists and is unique for the case with heterogeneous agents.
\begin{theorem}
The link formation game with heterogeneous agents has a unique CE. The CE is symmetric within each type, namely $a^{CE}_{k_i} = a^{CE}_k, \forall k_i, \forall k$. Moreover, $(a^{CE}_1,...,a^{CE}_K) \not\in \mathcal{C}$.
\end{theorem}
\begin{proof}
Based on our previous analysis on homogeneous epidemic networks, for each type $k$ agent, the long-term utility is
\begin{align}
U_{H,k}(a_{k_i}, \a_{-k_i}) = \frac{\rho + \delta_k}{\rho}\frac{b(a_{k_i})}{\rho + \delta_k + \beta_i(\a)} = \frac{\rho + \delta_k}{\rho}\frac{b(a_i)}{\rho + \delta_k + \beta\theta(\a)a_{k_i}}
\end{align}
Following the same analysis in the proof of Theorem 1, for any given $\theta > 0$, there is a unique maximizer for a given $\delta_k$. Therefore, agents of the same type must take the same action $a^*_k$, which is unique solution of
\begin{align}\label{BRheter}
\frac{b(a_k)}{b'(a_k)} - a_k = \frac{\rho + \delta_k}{\beta \theta}
\end{align}
We have also shown that LHS is increasing in $a$. Therefore, for a given $\theta$, a larger $\delta_k$ implies a larger $a^*_k$ and $a^*_k, \forall k$ decreases with $\theta$.

Now, the equilibrium must satisfy \eqref{SDheter}. Next, we show that there indeed exists a unique $\theta^*$ that satisfies both \eqref{BRheter} and \eqref{SDheter}. This is to prove that there is a solution $\theta$ to
\begin{align}
g(\theta) \triangleq \sum_k w_k \frac{ \beta \delta^{-1}_k a^*_k(\theta)}{\theta  \beta \delta^{-1}_k a^*_k(\theta) + 1} = 1
\end{align}
where $a^*_k(\theta)$ is determined by \eqref{BRheter}. It is easy to verify that $g(\theta)$ is monotonically decreasing in $\theta$ because $a^*_k(\theta)$ decreases with $\theta$. Moreover, $g(1) < \sum_k w_k =1$. Because as $\theta \to 0$, $a^*_k(\theta) \to \infty$, we have $\lim_{\theta \to 0} g(\theta) \to \infty$. As a result, there must be a unique $\theta \in (0,1)$ that is solution of $g(\theta) = 1$. Since $\theta$ is unique, the equilibrium link formation strategies must also be unique.
\end{proof}

\section{Numerical Results}
In this section, we provide numerical results to illustrate the various aspects of our framework and highlight the importance of understanding the strategic link formation behavior for designing effective protection mechanisms.

\begin{figure}
\centering
\begin{minipage}[b]{0.45\linewidth}
  \centering
  \includegraphics[width=1\textwidth]{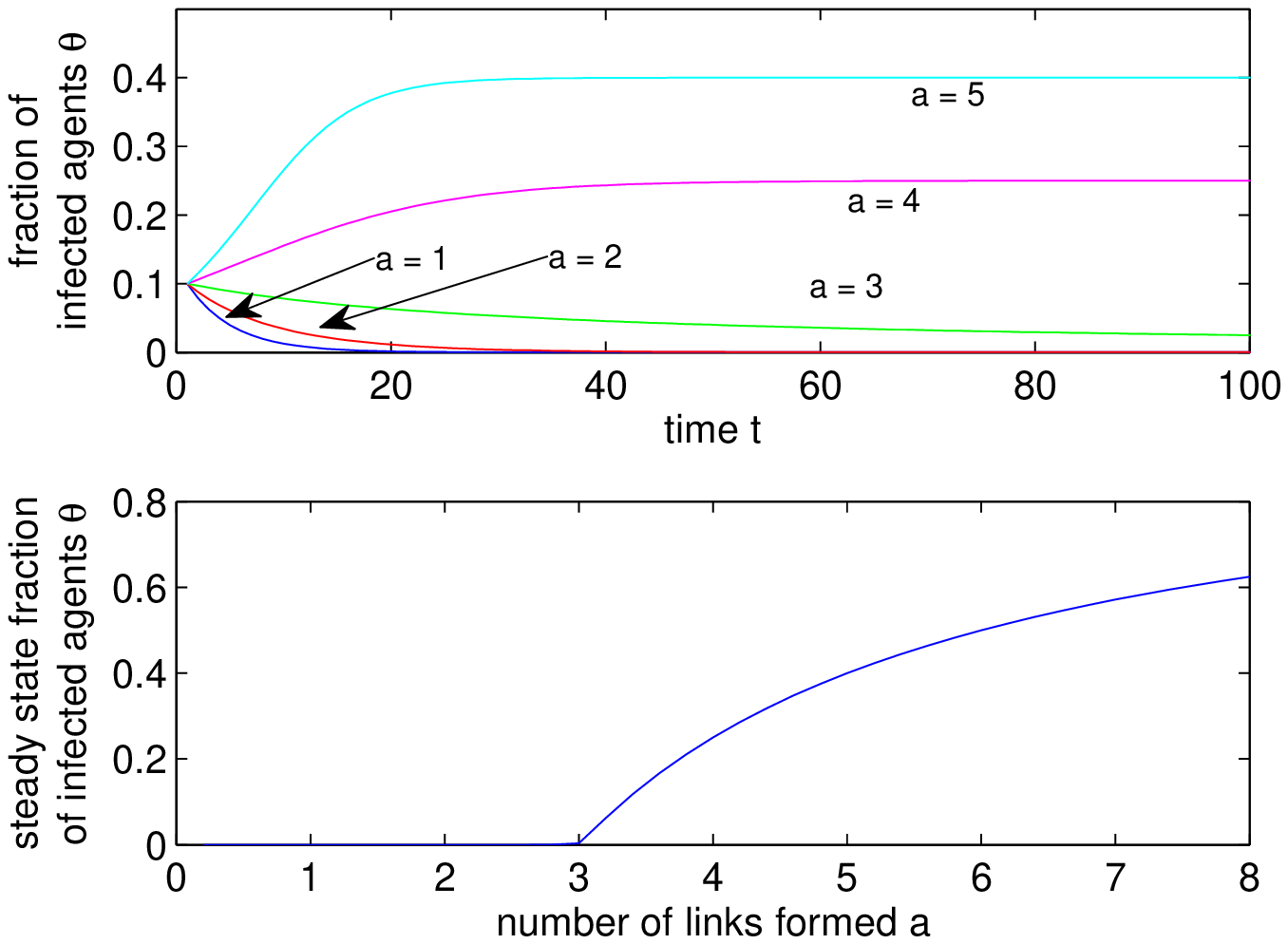}\\
  \caption{Steady state for non-strategic agents.}\label{constant}
\end{minipage}
\begin{minipage}[b]{0.45\linewidth}
  \centering
  \includegraphics[width=1\textwidth]{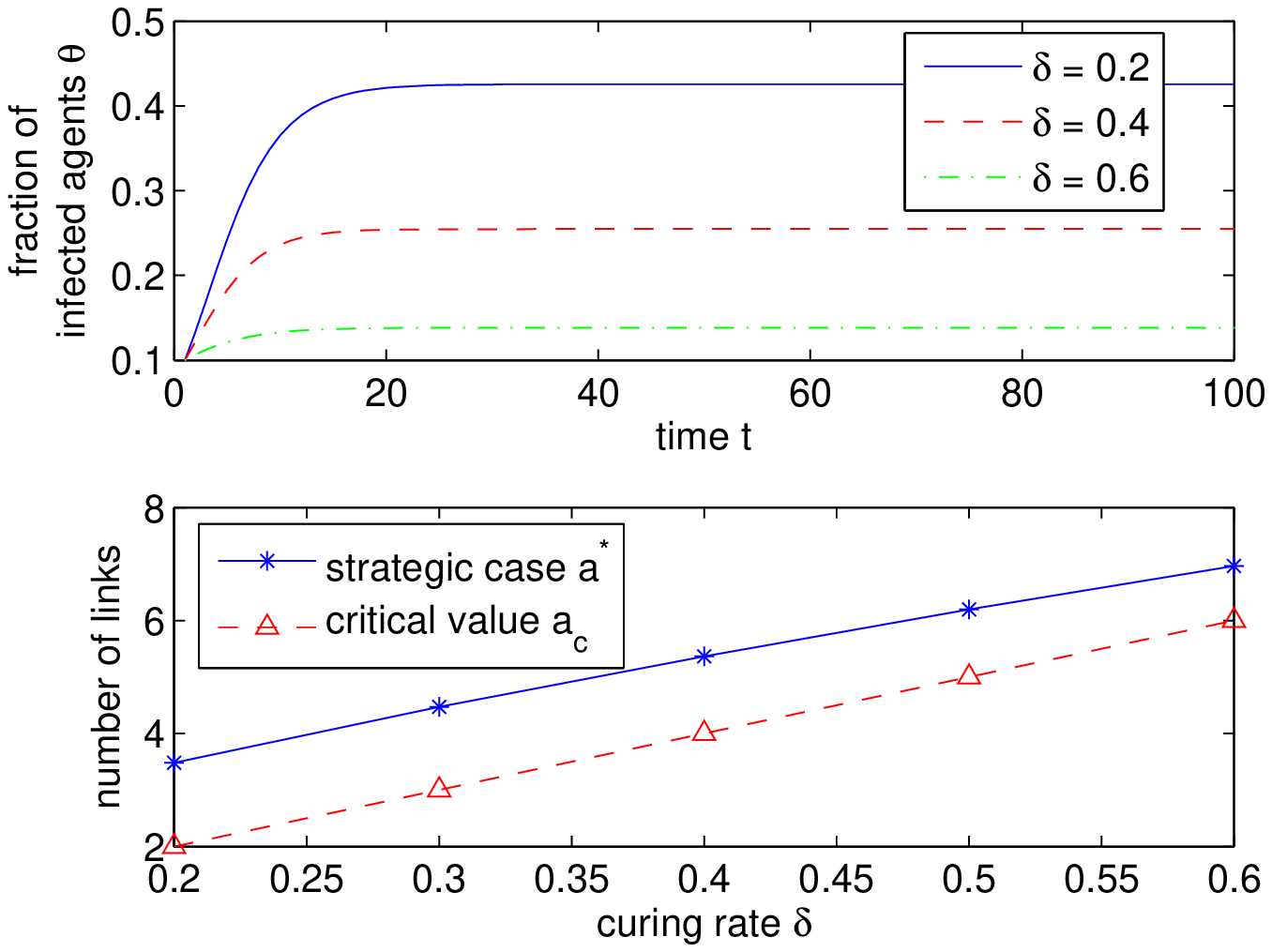}\\
  \caption{Steady state for strategic agents.}\label{strategic}
\end{minipage}
\end{figure}

First, we investigate the steady state behavior for both the fixed link formation strategy and the strategic link formation strategy (under system parameters $\beta = 0.1, \delta = 0.3, \rho = 0.05$). The upper plot in Figure \ref{constant} shows how the system state evolves over time for different fixed link formation strategies. Depending on the actions, the fraction of infected agents converges to different values. The lower plot in Figure \ref{constant} further illustrates the converged system state as a function of the link formation actions. As predicted by Proposition 1, there exists a threshold value of $a$ such that the infection dies out if the link formation action is below the threshold. Figure \ref{strategic} shows how the system dynamics changes when agents are strategic. The upper plot shows the evolution of the system state for various curing rates $\delta$. In each case, the system converges to a positive fraction of infected agents. This fraction decreases as the curing rate increases. The lower plot shows the link formation action $a^*$ in the steady state as a function of $\delta$. When the curing rate is larger, agents tend to form more links since infection is a smaller threat. However, the number of links that strategic agents are willing to form is always larger than the critical number in the non-strategic case.

\begin{figure}
\centering
\begin{minipage}[b]{0.45\linewidth}
  \centering
  \includegraphics[width=1\textwidth]{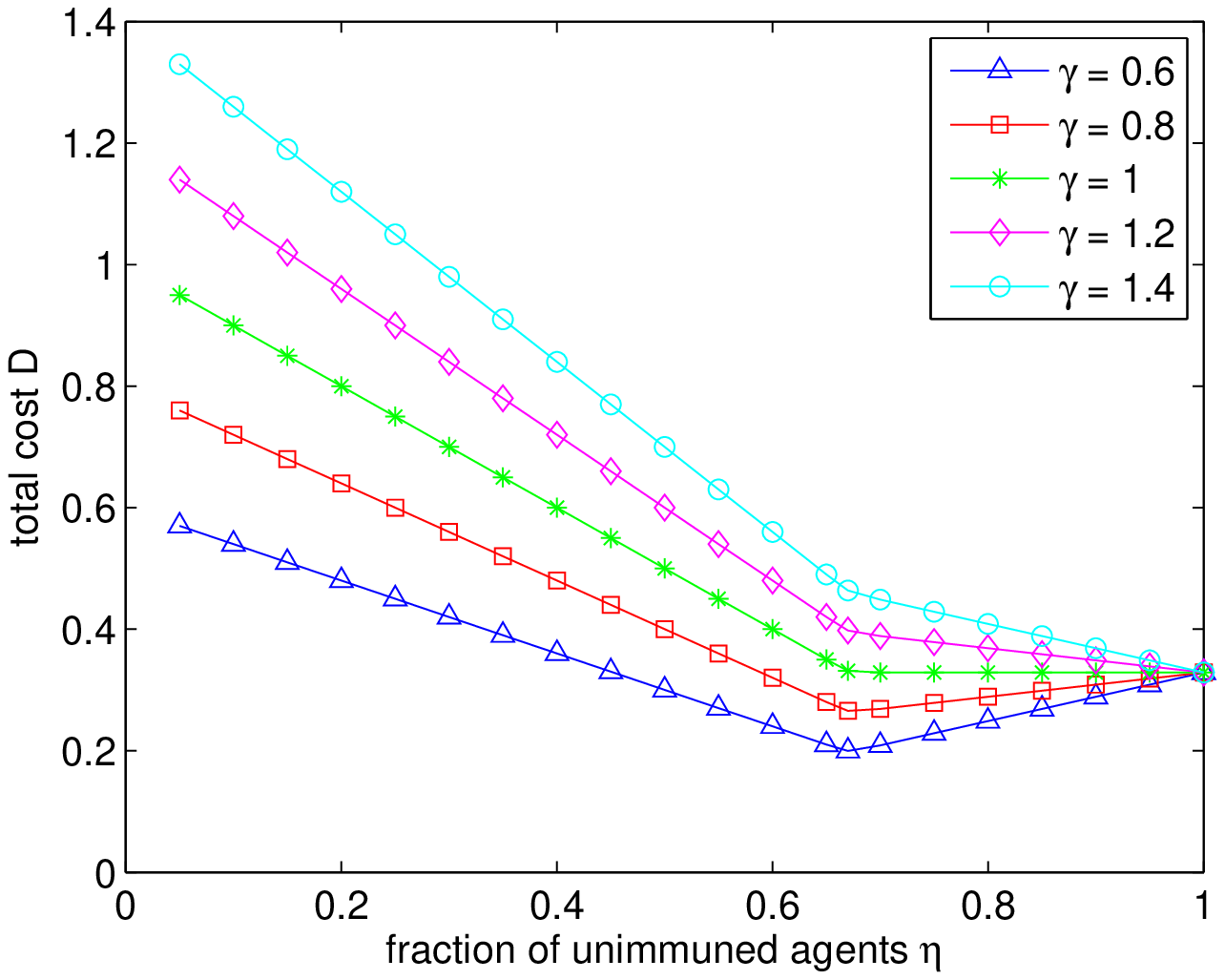}\\
  \caption{Optimal protection scheme for non-strategic agents.}\label{fixedprotection}
\end{minipage}
\begin{minipage}[b]{0.45\linewidth}
  \centering
  \includegraphics[width=1\textwidth]{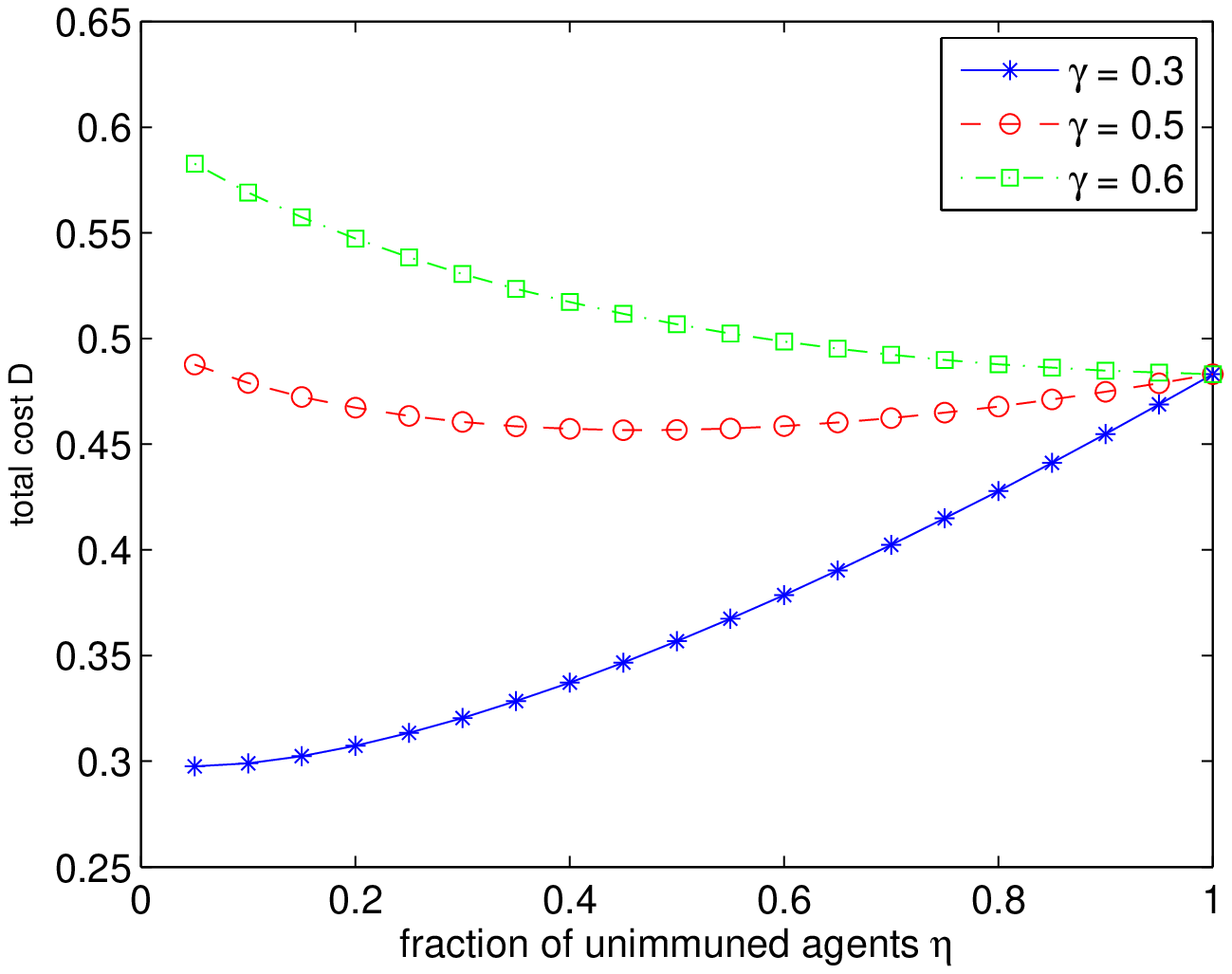}\\
  \caption{Optimal protection scheme for strategic agents.}\label{straprotection}
\end{minipage}
\end{figure}

Next, we show that different immunization protection schemes are needed when agents are strategic ($\beta = 0.1, \delta = 0.3, \rho = 0.05$). Figure \ref{fixedprotection} illustrates how the total cost $D$ varies as the immunization scheme changes for various values of the unit immunization cost $\gamma$. To enable direct comparisons, the action is fixed at $a = 4.47$, which is the same as the steady state action for strategic agents when $\delta = 0.3$ (see Figure \ref{strategic}). As we can see, when the unit immunization cost is large ($\gamma > 1$), the optimal scheme immunizes zero agents (i.e. $\eta = 1$) whereas when the unit immunization cost is small ($\gamma < 1$), the optimal scheme has $\eta = \frac{\delta}{\beta a} = 0.67$ regardless of the specific value of $\gamma$.  Thus, for a given fixed strategy, only two schemes ($\eta = 1$ or $\eta = \frac{\delta}{\beta a}$) can be the optimal scheme. Moreover, it is never optimal to immunize all agents in all cases. This is what Proposition 2 predicted. On the other hand, when agents are strategic, immunizing all agents is actually the optimal protection when the unit immunization cost is sufficiently low, as shown in Figure \ref{straprotection}. The figure also shows that depending on the specific value of $\gamma$, the optimal $\eta$ is different. When the unit immunization cost is too large, the optimal scheme is to immunize zero agents.

\begin{figure}
  \centering
  \includegraphics[width=0.6\textwidth]{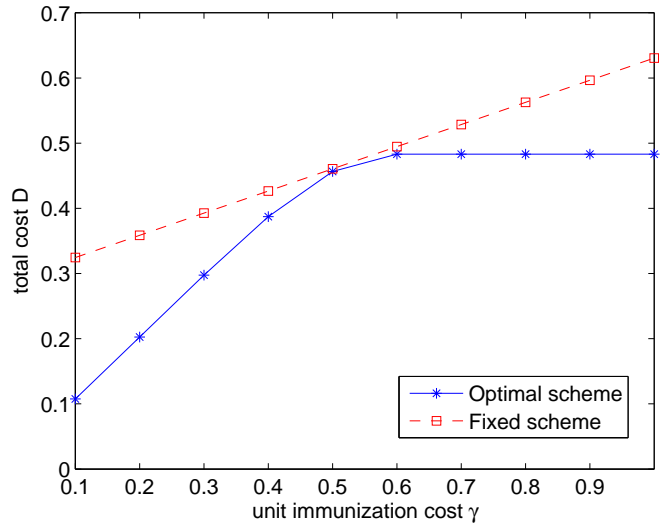}\\
  \caption{Performance loss due to wrong protection scheme.}\label{cost}
\end{figure}

The above comparison demonstrates the significant importance of understanding the strategic behavior of agents. Suppose that the system designer observes an average degree of 4.47 when there is no immunization protection and wrongly believes that agents are not strategic. Consequently, the designer will use an immunization scheme $\eta = 0.67$ when $\gamma = 0.3$ in order to minimize the total cost. However, since agents are actually strategic, the correct optimal immunization scheme should be $\eta = 0$. This leads to a significant performance loss. Using the curve for $\gamma = 0.3$ in Figure \ref{straprotection}, we can compute the increased cost, which is more than 30\% than that if the correct immunization scheme is used. Figure \ref{cost} illustrates the total cost difference for various values of the immunization cost. The cost incurred by using the fixed scheme is the closest to the optimal cost when $\gamma$ is around 0.5. This is because the fixed scheme happens to be close to the optimal scheme. However, when $\gamma$ is either larger or smaller, significant performance loss is incurred due to lack of understanding of the strategic behavior of agents. When $\gamma$ is small, the total system cost can increase by more than 3 times.

\begin{figure}
\centering
\begin{minipage}[b]{0.45\linewidth}
  \centering
  \includegraphics[width=1\textwidth]{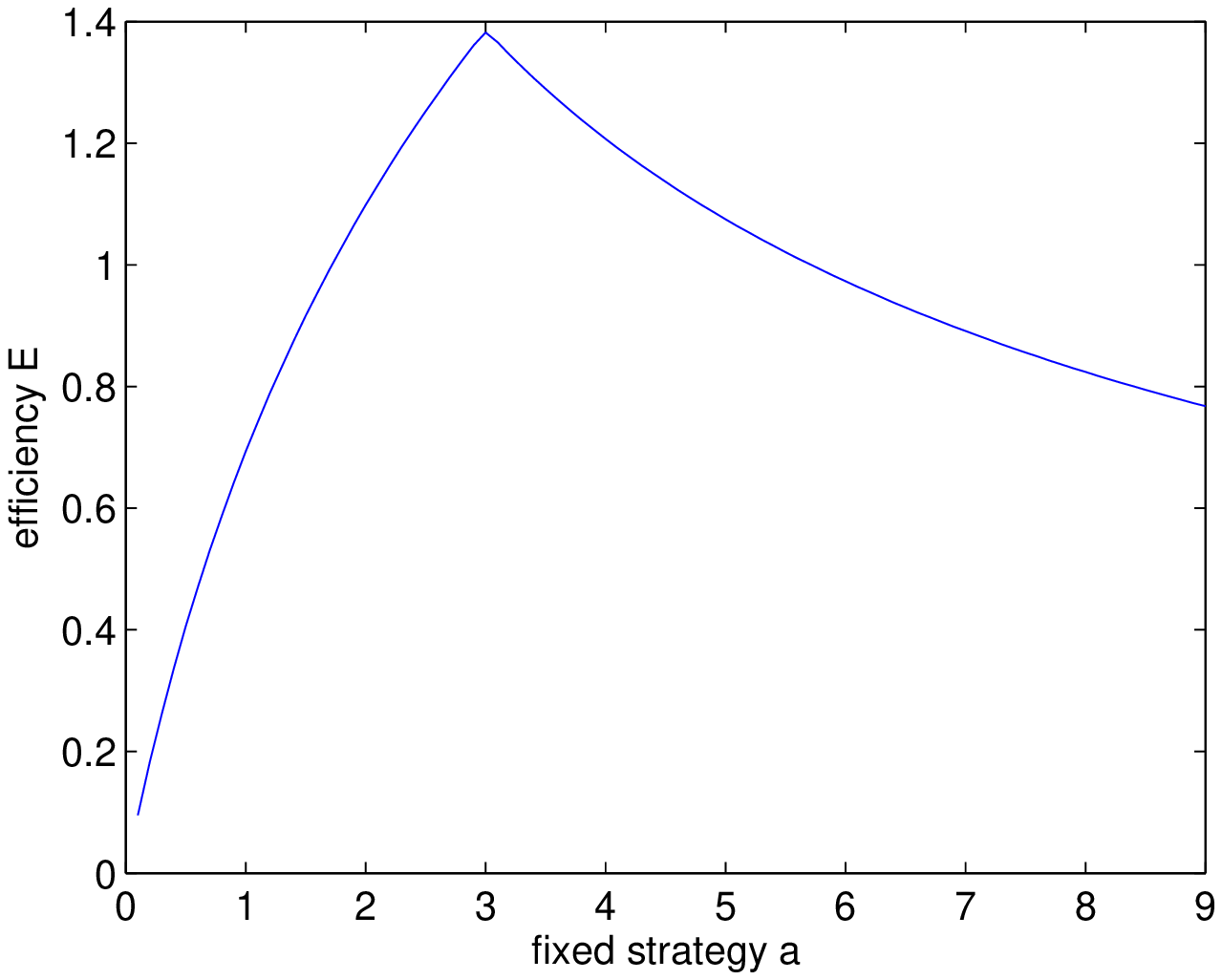}\\
  \caption{Social optimal efficiency.}\label{fixedefficiency}
\end{minipage}
\begin{minipage}[b]{0.45\linewidth}
  \centering
  \includegraphics[width=1\textwidth]{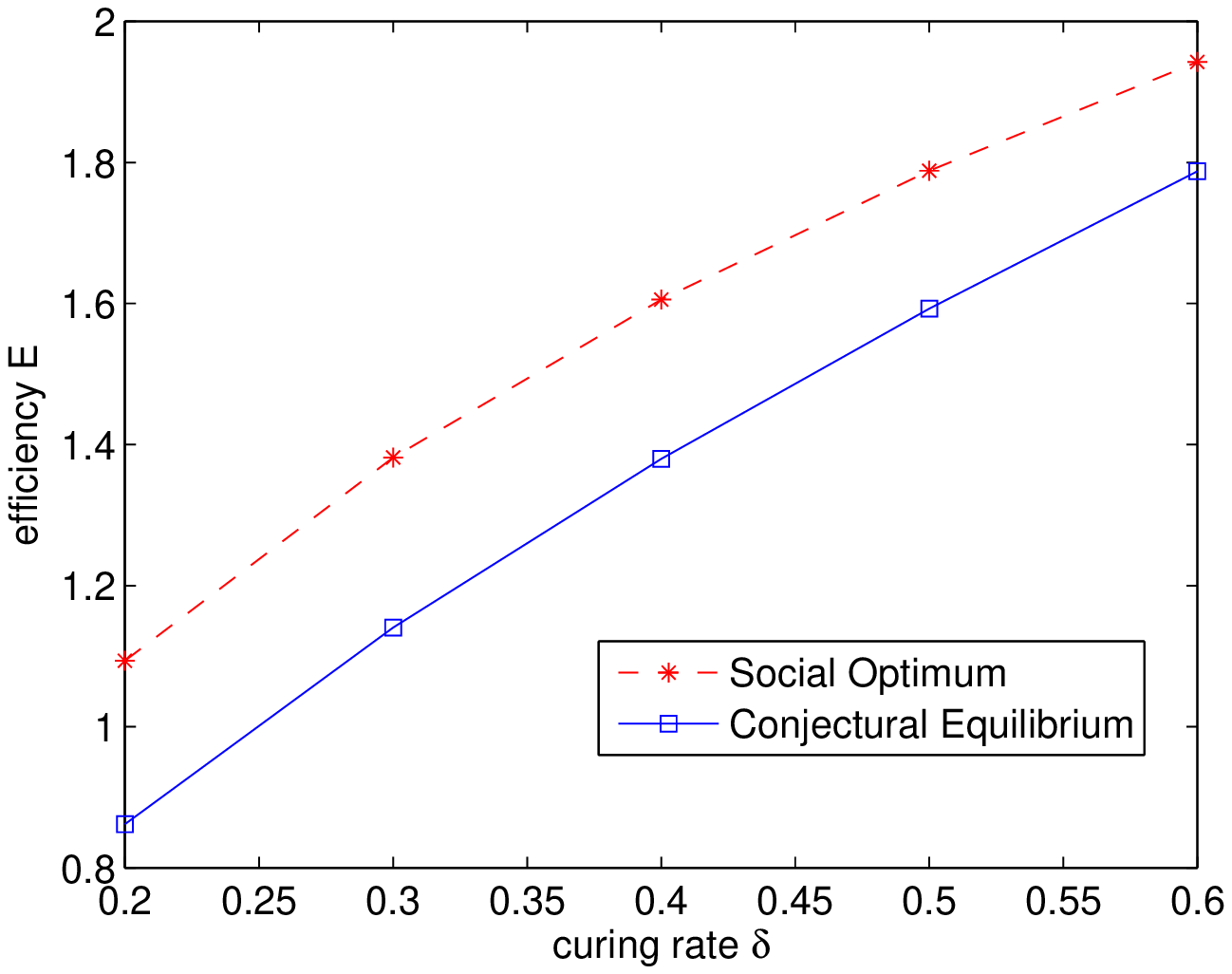}\\
  \caption{Efficiency loss due to strategic behavior.}\label{PoA}
\end{minipage}
\end{figure}

Finally, we investigate the achievable system efficiency. When agents are non-strategic, the system design can prescribe link formation actions to agents to maximize the system efficiency. Figure \ref{fixedefficiency} shows the achievable efficiency as a function of the link formation action ($\beta = 0.1, \delta = 0.3, \rho = 0.05$). As predicted in Proposition 3, the efficiency is maximized when the agents take the critical value $a_c = 3$. When agents are strategic, they tend to form higher than the critical number of links, thereby incurring efficiency loss. This is illustrated in Figure \ref{PoA} for various values of the curing rate $\delta$.

\section{Conclusions}
In this paper, we studied the epidemic dynamics in networks endogenously formed by strategic agents. We showed that such networks exhibit significantly different features than networks that are exogenously given. An important lesson learned from our analysis is that ignoring the strategic nature of agents in forming links may result in significantly increased system cost. Therefore, protection mechanisms that previously apply to non-strategic networks need to be re-designed and adjusted for networks formed by strategic agents.

Although this paper provides a number of key insights for designing network formed by strategic agents, we are keenly aware that there are also limitations in the current model, which tends to be simplistic and stylized. For instance, the current model focuses on only the average behavior of agents using an mean-field approximation. When there are only a limited number of agents who may be interacting with topological constraints, different formalism and analysis are required to take into account the identity of agents. In the current model, agents are strategic only in forming links. Another important future work direction is to understand the epidemic dynamics when agents jointly choose link formation and security investment actions.

% or


%\appendices
%% you can choose not to have a title for an appendix
%% if you want by leaving the argument blank
%\section{}
%% use section* for acknowledgement
%\section*{Acknowledgment}

\ifCLASSOPTIONcaptionsoff
  \newpage
\fi



\bibliographystyle{IEEEtran}
\bibliography{refs}

\end{document}